\documentclass[twocolumn,aps,pra,10pt,amsmath,amssymb,secnumarabic,%
     nofootinbib,superscriptaddress]{revtex4-1}

\usepackage{mathptmx,amsthm,tikz}
\providecommand{\tensor}{\otimes}
\providecommand{\affiliation}{\address}
\providecommand{\acknowledgments}{\section*{Acknowledgements}}

\newtheorem{theorem}{Theorem}
\newtheorem{proposition}[theorem]{Proposition}
\newtheorem{corollary}[theorem]{Corollary}
\theoremstyle{remark}
\newtheorem*{remark}{Remark}

\newcommand{\C}{\mathbb{C}}
\newcommand{\CP}{\mathbb{CP}}
\newcommand{\cH}{\mathcal{H}}
\newcommand{\Sym}{\mathrm{Sym}}
\newcommand{\unif}{{\mathrm{unif}}}
\newcommand{\GL}{{\mathrm{GL}}}

\renewcommand{\tensor}{\otimes}
\newcommand{\ket}[1]{|#1\rangle}
\newcommand{\bra}[1]{\langle #1|}
\newcommand{\defeq}{\stackrel{\mathrm{def}}{=}}

\newcommand{\multiset}[2]{\left(\kern-.15em\binom{#1}{#2}\kern-.15em\right)}

\newcommand{\ie}{\textit{i.e.}}

\newcommand{\thm}[1]{Theorem~\ref{#1}}

\newcommand{\cor}[1]{Corollary~\ref{#1}}
\newcommand{\eq}[2]{\begin{equation}\label{#1}#2\end{equation}}

\begin{document}
\title{The bosonic birthday paradox}

\author{Alex Arkhipov}
\thanks{Supported by an Akamai Foundation Fellowship.}
\affiliation{Massachusetts Institute of Technology}
\author{Greg Kuperberg}
\thanks{Partly supported by NSF grant CCF-1013079.}
\affiliation{University of California, Davis}

\begin{abstract} We motivate and prove a version of the birthday paradox
for $k$ identical bosons in $n$ possible modes.  If the bosons are in the
uniform mixed state, also called the maximally mixed quantum state, then
we need $k \sim \sqrt{n}$ bosons to expect two in the same state, which
is smaller by a factor of $\sqrt{2}$ than in the case of distinguishable
objects (boltzmannons).  While the core result is elementary, we generalize
the hypothesis and strengthen the conclusion in several ways.  One side
result is that boltzmannons with a randomly chosen multinomial distribution
have the same birthday statistics as bosons.  This last result is
interesting as a quantum proof of a classical probability theorem;
we also give a classical proof.
\end{abstract}

\maketitle

The traditional birthday paradox says that given a calendar with $n$
days, there is a significant chance (bounded away from 0) that a room
with $\Omega(\sqrt{n})$ people with uniformly random birthdays has two
with the same birthday.  Aaronson and Arkhipov \cite{AA:optics} discuss
the same paradox for randomly chosen bosons.  Here we present a different
treatment of the same problem. In fact we will present two ``paradoxes".
The first result (which Aaronson and Arkhipov derived, in a less general
form) is that although bosons prefer to have the same birthday, they
have the same asymptotic behavior in the birthday problem, up to constant
factors, as distinguishable particles (boltzmannons).  The second result
is that they have exactly the same behavior, non-asymptotically, as $n$
i.i.d. boltzmannons whose common distribution is a randomly chosen point in
the simplex of all distributions on $n$ configurations.  This leads to an
interesting result in classical probability with a quantum probability proof.

We assume that the Hilbert space for one particle is $\cH = \C^n$.
We assume a self-adjoint birthday operator
$$B:\cH \to \cH$$
with eigenvalues $1,2,\dots,n$ in some basis.  The Hilbert space of $k$
bosons is then the symmetric power
$$S^k(\cH) \cong \C^{\multiset{n}{k}},$$
using the multiset coefficient notation
\eq{e:multisetdef}{\multiset{n}{k} \defeq \binom{n+k-1}{k}.}
In the terminology used for identical particles, the states of a basis of
$\cH$ are called \emph{modes}.

In the traditional version of the classical birthday problem, we assume
the uniform distribution on all $n^k$ choices of the birthdays of the
$k$ people.  The uniform distribution $\mu_\unif(X)$ on any finite set $X$
can be characterized in either of two ways:  It is the unique distribution
with the most entropy, $\log |X|$; and the unique distribution with the
most symmetry, $\Sym(X)$.

We will consider an analogue of the uniform distribution for a quantum
system with a Hilbert space $\cH$:  the mixed state $\rho_\unif(\cH)$
whose density matrix is the scaled identity on $\cH$.  Like the classical
state $\mu_\unif$, the quantum state $\rho_\unif(\cH)$ is the unique
state on $\cH$ with the most entropy, $\log \dim \cH$; and the unique
state with the most symmetry, $U(\cH)$.  Moreover, $\rho_\unif(\cH)$
is the unique state that yields the distribution $\mu_\unif(X)$ for any
complete measurement that takes values in a set $X$.

We will use the uniform state $\rho_\unif = \rho_\unif(S^k(\cH))$ on the
joint Hilbert space of $k$ bosons.  Then, the measurement $S^k(B)$ of all
birthdays of $\rho_\unif$ yields the uniform distribution $\mu_\unif$
on configurations of $k$ \emph{unlabelled} people with $n$ possible
birthdays.  (It is also standard to refer to unlabelled balls in labelled
boxes, but we will stick to the birthday metaphor.)  Moreover, this
particular uniform state can be justified using less symmetry than the
largest available unitary group $U(S^k(\cH))$:

\begin{proposition} The state $\rho_\unif$ on $S^k(\cH)$ is the unique
state which is invariant under the unitary group $U(\cH)$.
\label{p:symmetry} \end{proposition}

\begin{proof} Suppose that $\rho$ is a $U(\cH)$-invariant state on
$S^k(\cH)$, \ie, a $U(\cH)$-invariant density operator.  Schur's Lemma
says that if $V$ is an irreducible complex representation of a group $G$,
then every $G$-invariant operator on $V$ is proportional to the identity.
Thus it is sufficient (and also necessary, if either $V$ is unitary or $G$ is
compact) for $V$ to be irreducible. It is a standard fact of representation
theory \cite[\S6.1]{FH:gtm} that $S^k(\C^n)$ is an irreducible representation
of $\GL(n,\C)$.  It is another standard fact \cite[\S26.1]{FH:gtm} that
$\GL(n,\C)$ and $U(n)$ have the same irreducible representations,
since the former is the complexification of the latter.
\end{proof}

This symmetry implies that $\rho_\unif$ is the $U(\cH)$-average of any
state, since such an average must be invariant with respect to the action
of $U(\cH)$.

\begin{corollary} Putting $k$ bosons in any state $\sigma$ on $S^k(\cH)$,
and then applying a Haar-random unitary matrix in $U(\cH)$ yields
the state $\rho_\unif$.  \label{c:average} \end{corollary}

Aaronson and Arkhipov consider such an average for a particular choice of
$\sigma$, where $\sigma$ is the pure state
$$\ket{\psi} = \ket{1,2,3,\ldots,k}$$
in which the $k$ bosons are in distinct modes (which requires $k \le n$).
Another choice considered below is
$$\ket{\psi} = \ket{1,1,1,\ldots,1}$$
in which the bosons are all in the same mode.  There are many choices for
$\sigma$, but \cor{c:average} says that they all become the same when they
are averaged.

We will now look at the asymptotics of $j$-fold birthdays in $\rho_\unif$.
We will use the notation $f(n) \sim g(n)$ to mean that $f(n)/g(n) \to 1$,
or equivalently that $f(n) = g(n)(1+o(1))$.

\begin{theorem} Suppose that there are $k$ bosons with $n$ modes, suppose
that they are in the uniform state $\rho_\unif$, and suppose that $k
\sim cn^{(j-1)/j}$ as $n \to \infty$, for some integer $j \ge 2$ and
some constant $c > 0$.  Then the number of $j$-fold birthdays converges
in distribution to a Poisson random variable with mean $c^j$, while the
number of $(j+1)$-fold-or-more birthdays converges to 0.
\label{th:bosons} \end{theorem}

This is the same asymptotic answer as in the case of boltzmannons, except
that the mean in that case is $c^j/j!$.  In fact, our argument in the case
of bosons is very similar to a standard argument in the case of boltzmannons.

\begin{proof} Recall that the joint measurement $S^k(B)$ of all of the
birthdays yields the uniform distribution on $k$ unlabelled people among $n$
calendar days.  The probability that the first birthday has at least $j+1$
people is
$$\multiset{n}{k-j-1}\biggm/\multiset{n}{k} \sim \frac{k^{j+1}}{(n+k)^{j+1}}$$
for fixed $j$ and $n, k \gg 1$.  Taking $k = O(n^{(j-1)/j})$ and
summing over all $n$ days, the expected number of $(j+1)$-fold-or-more
birthdays is $O(n^{-1/j})$, which vanishes as $n \to \infty$.

Meanwhile the probability that the first $\ell$ days each have at least $j$
people is
$$\multiset{n}{k-j\ell}\biggm/\multiset{n}{k} =
    \prod_{a = 0}^{j\ell-1} \frac{k-a}{n+k-a}
    \sim \frac{k^{j\ell}}{(n+k)^{-j\ell}},$$
where the approximation holds for fixed $j$ and $\ell$ and $n, k \gg 1$.
Summing over all $\binom{n}{\ell} \sim \frac{n^\ell}{\ell!}$ choices of
the $\ell$ days, we obtain that if $X$ is a random variable representing the
number of $j$-fold birthdays, then
$$E\left[ \binom{X}{\ell} \right] \sim \frac{c^{j\ell}}{\ell!}.$$
So in the limit, the $\ell$th factorial moment is $c^{j\ell}$, which the
same answer in the limit as a Poisson random variable with mean $c^j$.
To conclude the argument, the Poisson distribution is determined by its
moments.
\end{proof}

The calculation for the narrow question of the probability of at least one
repeated birthday is simpler.  The probability that all of the birthdays
are distinct is
$$\binom{n}{k}\biggm/\multiset{n}{k} =
\prod_{a=0}^{k-1} \frac{1-\frac{a}{n}}{1+\frac{a}{n}} \sim e^{-k^2/n}$$
as long as $k = o(n^{3/4})$.  The approximation is established by taking
the logarithm of both sides and then applying the Taylor series estimate
$$\ln \frac{1-x}{1+x} = -2x + O(x^3).$$

\begin{corollary} For $n$ modes, we need $k \sim \sqrt {n \ln 2}$ bosons
to expect a repeated birthday with majority probability.
\label{c:birthday} \end{corollary} 

This differs by only a constant factor from the $k \sim \sqrt {2 n \ln
2}$ people needed to expect a repeated birthday in the classical birthday
problem with distinguishable people.

\begin{remark} We should say something about independent but non-uniform
bosons.  The notion of independence for bosons is subtle.  One reasonable
and widely used notion is to first choose a distribution $\mu$ for the
birthdays of one boson, and to model it by a diagonal density matrix in
the birthday basis.  Then there is a unique distribution on $k$ bosons such
that if $k-1$ of the bosons are fixed, the conditional distribution of the
last one is given by $\mu$.  This distribution is also a thermal state,
also known as a Maxwell-Gibbs state, for non-interacting bosons.  It was
discovered by Bose and Einstein that under fairly mild assumptions on $\mu$,
almost all of the bosons have the most likely birthday.  This paradox
is commonly known as Bose-Einstein condensation.
\end{remark}

\cor{c:average} implies an interesting second model for the joint
distribution of birthdays of $k$ bosons.

\begin{theorem} The joint birthday distribution of $k$ bosons in the uniform
state $\rho_\unif$ is identical to the average of $k$ i.i.d. boltzmannons,
if their common distribution is given by a uniformly random point in the
simplex of distributions on the $n$ birthdays.
\label{th:simplex} \end{theorem}

By combining with the induced uniform distribution on the birthday
measurement, we obtain a corollary of \thm{th:simplex} that equates two
distributions in classical probability.

\begin{corollary} Consider a town in which all families first agree to
have children according to a common distribution on the days of the year,
which itself is chosen uniformly from the simplex of all distributions.
Then the children's birthdays behave as if the children were unlabelled,
\ie, if we make a table that only gives the number of children born on
each day, then all such tables are equally likely.
\label{c:simplex} \end{corollary}

In other words, the uniform average of all multinomial distributions on
multisubsets of size $k$ in a set of size $n$, is the uniform distribution
on multisubsets.

\begin{proof}[Proof of \thm{th:simplex}] Recall that the Hilbert space
of $k$ boltzmannons is $\cH^{\tensor k}$. Consider the state $\sigma =
(\ket{\psi}\bra{\psi})^{\tensor k}$, first for some fixed choice of
$\ket{\psi} \in \cH$.  This $\sigma$ yields independently distributed
birthdays for the $k$ boltzmannons, and the distribution of each one is
given by the measurement of one copy of $\ket{\psi}$.  Meanwhile, $\sigma$
is evidently a pure symmetric state, which means that these boltzmannons are
also bosons.  By \cor{c:average}, the average of all choices of $\sigma$,
with respect to Haar measure on $U(\cH)$, is the bosonic state $\rho_\unif$.

The Haar distribution of $\ket{\psi}$, or equivalently one column of a
matrix in $U(n)$, is given by Haar measure on the manifold of pure states
$\CP^{n-1}$.  The induced distribution of the birthday measurement is
given by the moment map
$$m:\CP^{n-1} \to \Delta_{n-1}$$
to the simplex of distributions on $n$ configurations \cite[\S
6.4]{Cannas:handbook}.  This moment map preserves normalized measure
\cite[\S 6.6]{Cannas:handbook}.  Thus a random choice of $\sigma$ amounts
to a random distribution on each birthday, drawn uniformly from the simplex
of distributions.  This establishes the claim of the theorem.
\end{proof}

\thm{th:simplex} yields a quantum proof of a classical probability result,
\cor{c:simplex}.  We also obtained a classical proof of the same result.

\begin{proof}[Classical proof of \cor{c:simplex}] The argument uses a
variation of the stars-and-bars notation for multisets \cite{Feller:intro}
that is also used to prove the identity \eqref{e:multisetdef}.  Namely,
we write a star for each of the $k$ children, with $n-1$ separating bars
between the $n$ calendar days.  For example, if there are $k=4$ children
and $n=6$ birthdays, then one possible choice for all of the 
birthdays is
$$\star \star | \star || \star ||,$$
in which two children are born on the first day, one on the second day, one
on the fourth day, and none on the other days.  We first choose locations
of $n-1$ bars independently and uniformly on the unit interval $I=[0,1]$.
This separates the interval into $n$ subintervals of length
$$p_1 + p_2 + \ldots + p_n = 1,$$
and we claim that the lengths of these subintervals are given by a
uniformly random point in the simplex of distributions.  (Because, if
we first take the bars to be numbered, they are distributed according to
uniform measure on $[0,1]^{n-1}$.  Then, erasing the numbers yields the
quotient $[0,1]^{n-1}/S_{n-1}$, which is a simplex and also has uniform
measure.  Then, taking the differences of successive points to obtain
the probabilities $p_j$ is a linear isomorphism, which also preserves
uniform measure.)  Then, if each child's birth is represented by a star
which is also at a uniformly random position in $[0,1]$, the probability
of the $j$th birthday is exactly $p_j$, the length of the $j$th interval.

We note that the ordering of the stars and bars determines the number of
children with each birthday.  We claim that these multiset choices are all
equally likely, as if the children had been bosons (with no distinguishing
state other than the date of birth).  This is made clear if we equivalently
choose $n-1+k$ points independently from $I$ all at once, and then choose
a random subset of $n-1$ points to be the bars and the other $k$ points to
be the stars.   These $\multiset{n}{k} = \binom{n-1+k}{k}$ equally likely
choices exactly correspond to a multiset choice of $k$ unlabelled children
distributed among $n$ days, as claimed.
\end{proof}

We conclude with a version of the birthday paradox for fermions.

\begin{theorem}[Pauli] Given $k$ fermions in any state on the exterior power
$\Lambda^k(\cH)$, there is no chance that any two have the same birthday.
\label{th:fermion} \end{theorem}

We leave the question of an anyonic birthday paradox, including non-abelian
anyons, as a topic for future work.

\vspace{\baselineskip}

\acknowledgments

The authors would like to thank Scott Aaronson for suggesting the problem.

\bibliographystyle{hamsplain}

% \bibliography{bib/qp,bib/books}

\providecommand{\bysame}{\leavevmode\hbox to3em{\hrulefill}\thinspace}
\providecommand{\MR}{\relax\ifhmode\unskip\space\fi MR }
\providecommand{\MRhref}[2]{%
  \href{http://www.ams.org/mathscinet-getitem?mr=#1}{#2}
}
\providecommand{\href}[2]{#2}

\end{document}